\documentclass[a4paper]{article}
\usepackage{geometry}                
\usepackage{amssymb,amsmath,amsthm}
\usepackage{url}
\usepackage{graphicx}
\usepackage{enumitem}
\usepackage{relsize}

\usepackage{ifxetex}
\ifxetex
\usepackage{fontspec,xltxtra,xunicode}
\usepackage{unicode-math}
\defaultfontfeatures{Mapping=tex-text}
\newfontfamily\cyrillicfont{STIX Two Text}
\setromanfont[Mapping=tex-text]{STIX Two Text}
\usepackage{polyglossia}
\setdefaultlanguage{english}
\setotherlanguage{russian}
\else
\usepackage{mathptmx}
\fi

\DeclareMathOperator{\KS}{\mathrm{C}\mskip 1.2mu}
\DeclareMathOperator{\KP}{\mathrm{K}\mskip 1.2mu}
\DeclareMathOperator{\mm}{\mathsf{m}\mskip 1.2mu}
\DeclareMathOperator{\ee}{\mathsf{e}\mskip 1.2mu}
\newcommand{\cnd}{\mskip 1mu|\mskip 1mu}
\let\le=\leqslant
\let\ge=\geqslant
\let\eps=\varepsilon
\newtheorem*{lemma*}{Lemma}
\newtheorem{lemma}{Lemma}
\newtheorem{proposition}{Proposition}
\newtheorem{theorem}{Theorem}
\theoremstyle{remark}
\newtheorem{remark}{Remark}
\newcommand{\Epsilon}{\mathcal{E}}
\newcommand{\Emax}{E_1}

\newcommand{\moreHorizSpaceInFormulas}{         
  \medmuskip=6mu plus 2mu minus 3mu
  \thickmuskip=8mu plus 2mu minus 3mu
}

\begin{document}
\title{Information Distance Revisited}
\author{Bruno Bauwens\footnote{
  National Research University Higher School of Economics, Faculty of Computerscience, 11 Pokrovsky Boulevard, Moscow, Russia
  \newline
  This work was initiated by Alexander Shen who informed the author about the error in~\cite{mahmud} during a discussion of the paper~\cite{vitanyi2017}.
  After explaining the proof of Theorem~\ref{th:main}, he simplified it and wrote the current manuscript. 
  Later, Theorem~\ref{th:informationDistanceEquality}, but he could not yet 
  verify the complete technical argument. 
  I am very grateful for his permission to submit the current document.
  The author is grateful to Mikhail Andreev for the proof of Proposition~\ref{prop:triangle_inequality} and many useful discussions.
  The author is also grateful to the participants of the Kolmogorov seminar in Moscow state university for useful discussions.
  }
}

\maketitle

\begin{abstract}
We consider the notion of information distance between two objects $x$ and $y$ 
introduced by Bennett, G\'acs, Li, Vitanyi, and Zurek~\cite{bglvz} 
as the minimal length of a program that computes $x$ from $y$ as well as computing $y$ from $x$, 
and study different versions of this notion. 
In the above paper, it was shown that the prefix version of information distance equals $\max(\KP(x\cnd y),\KP(y\cnd x)$ 
up to additive logarithmic terms. 
It was claimed by Mahmud~\cite{mahmud} that this equality holds up to additive $O(1)$-precision. 
We show that this claim is false, but does hold if the distance is at least logarithmic.
This implies that the original definition provides a metric on strings that are at superlogarithmically separated.
\end{abstract}

\section{Introduction}

Informally speaking, Kolmogorov complexity measures the amount of information in an object (say, a bit string) in bits. The complexity $\KS(x)$ of $x$ is defined as the minimal bit length of a program that generates $x$. This definition depends on the programming language used, but one can fix an optimal language that makes the complexity function minimal up to an $O(1)$ additive term. In a similar way one can define the \emph{conditional} Kolmogorov complexity $\KS(x\cnd y)$ of a string $x$ given some other string $y$ as a condition. Namely, we consider the minimal length of a program that transforms $y$ to $x$. Informally speaking, $\KS(x\cnd y)$ is the amount of information in $x$ that is missing in $y$, the number of bits that we should give in addition to~$y$ if we want to specify $x$.

The notion of \emph{information distance} was introduced in~\cite{bglvz} as ``the length of a shortest binary program that computes $x$ from  $y$ as well as computing $y$ from $x$.'' It is clear that such a program cannot be shorter than $\KS(x\cnd y)$  or $\KS(y\cnd x)$ since it performs both tasks; on the other hand, it cannot be much longer than the sum of these two quantities (we can combine the programs that map $x$ to $y$ and vice versa with a small overhead needed to separate the two parts and to distinguish $x$ from $y$). As the authors of~\cite{bglvz} note, ``being shortest, such a program should take advantage of any redundancy between the information required to go from $x$ to $y$ and the information required to go from $y$ to~$x$'', and the natural question arises: to what extent is this possible? The main result of~\cite{bglvz} gives the strongest upper bound possible and says that the information distance equals
$
a b 
$
with logarithmic precision. In many applications, this characterization turned out to be useful [some ref].
In fact, in~\cite{bglvz} the prefix version of complexity, denoted by $\KP(x\cnd y)$, and the corresponding definition of information distance were used; see, e.g.~\cite{usv} for the detailed explanation of different complexity definitions. 
The difference between prefix and plain versions is logarithmic, 
so it does not matter whether we use plain or prefix versions if we are interested in results with logarithmic precision. 
However, the prefix version of the above characterization has an advantage: 
after adding a large enough constant, this distance satisfies the triangle inequality. 
The plain variant does not have this property, and this follows from Proposition~\ref{prop:triangle_inequality} below.
However, several inequalities that are true with logarithmic precision for plain complexity become true with $O(1)$-precision if prefix complexity is used. So one could hope that a stronger result with $O(1)$-precision holds for prefix complexity. 
One could hope that a similar result with $O(1)$-precision also holds for prefix complexity. If this is true, then also the original definition satisfies the triangle inequality (after a constant increase). Such a claim was indeed made in~\cite{mahmud}; in~\cite{lzlm} a similar claim is made with reference to~\cite{bglvz}.\footnote{The authors of~\cite{lzlm} define (section 2.2) the function $E(x,y)$ as the prefix-free non-bipartite version of the information distance (see the discussion below in section~\ref{subsec:four-versions}) and then write: ``the following theorem proved in [4] was a surprise: Theorem 1. $E(x,y)=\max\{\KS(x\cnd y),\KS(y\cnd x)\}$''. They do not mention that in the paper they cited as [4] (it is~\cite{bglvz} in our list) there is a logarithmic error term; in fact, they do not mention any error terms (though in other statements the constant term is written explicitly). Probably this is a typo, since more general Theorem 2 in~\cite{lzlm} does contain a logarithmic error term.} Unfortunately, the proof in~\cite{mahmud} contains an error and (as we will show) the result is not valid for prefix complexity with $O(1)$-precision. On the other hand, it is easy to see that the original argument from~\cite{bglvz} can be adapted for plain complexity to obtain the result with $O(1)$-precision, as noted in~\cite{vitanyi2017}.

In this paper we try to clarify the situation and discuss the possible definitions of information distance in plain and prefix versions, and their subtle points (one of these subtle points was the source of an error in~\cite{mahmud}). We also discuss some related notions. In Section~\ref{sec:plain} we consider the easier case of plain complexity; then in Section~\ref{sec:prefix} we discuss the different definitions of prefix complexity (with prefix-free and prefix-stable machines, as well as definitions using the a priori probability) and in Section~\ref{sec:prefix-distance} we discuss their counterparts for the information distance. In Section~\ref{sec:game} we use the game approach to show that indeed the relation between information distance (in the prefix version) and conditional prefix complexity is \emph{not} valid with $O(1)$-precision, contrary to what is said in~\cite{mahmud}. Finally, we show that if the information distance is at least logarithmic, then equality holds.

\section{Plain complexity and information distance}\label{sec:plain}

Let us recall the definition of plain conditional Kolmogorov complexity. Let $U(p,x)$ be a computable partial function of two string arguments; its values are also binary strings. We may think of $U$ as an interpreter of some programming language. The first argument $p$ is considered as a program and the second argument is an input for this program. Then we define the complexity function
\[
\KS_U (y\cnd x) = \min\{|p|\colon U(p,x)=y\};
\]
here $|p|$ stands for the length of a binary string $p$, so the right hand side is the minimal length of a program that produces output $y$ given input~$x$. The classical Solomonoff--Kolmogorov theorem says that there exists an optimal $U$ that makes $\KS_U$ minimal up to an $O(1)$-additive term. We fix some optimal $U$ and then denote $\KS_U$ by just $\KS$. (See, e.g., \cite{lv,usv} for the details.)

Now we want to define the information distance between $x$ and $y$. One can try the following approach: take some optimal $U$ from the definition of conditional complexity and then define
$$
 E_U(x,y)=\min\{|p|\colon U(p,x)=y \text{ and } U(p,y)=x\},
$$
i.e., consider the minimal length of a program that both maps $x$ to $y$ and $y$ to $x$. However, there is a caveat, as the following simple observation shows.

\begin{proposition}
There exists some computable partial function $U$ that makes $\KS_U$ minimal up to an $O(1)$ additive term, and still $E_U(x,y)$ is infinite for some strings $x$ and $y$ and therefore not minimal.
\end{proposition}

\begin{proof}
Consider an optimal function $U$ and then define $U'$ such that $U(\Lambda,x)=\Lambda$  where $\Lambda$ is the empty string, $U'(0p,x)=0U(p,x)$ and $U'(1p,x)=1U(p,x)$. In other terms, $U'$ copies the first bit of the program to the output and then applies $U$ to the rest of the program and the input. It is easy to see that $\KS_{U'}$ is minimal up to an $O(1)$ additive term, but $U'(q,\cdot)$ has the same first bit as $q$, so if $x$ and $y$ have different first bits, there is no $q$ such that $U(q,x)=y$ and $U(q,y)=x$ at the same time.
\end{proof}

On the other hand, the following proposition is true (and can be proven in the same way as the existence of the optimal $U$ for conditional complexity):

\begin{proposition}
There exists a computable partial function $U$ that makes $E_U$ minimal up to $O(1)$ additive term.
\end{proposition}

Now we may define information distance for plain complexity as the minimal function~$E_U$. It turns out that the original argument from~\cite{bglvz} can be easily adapted to show the following result (that is a special case of a more general result about several strings proven in~\cite{vitanyi2017}):

\begin{theorem}\label{thm:bglvz}
The minimal function $E_U$ equals $\max(\KS(x\cnd y),\KS(y\cnd x))+O(1)$.
\end{theorem}

\begin{proof}
We provide the adapted proof here for the reader's convenience.  In one direction we have to prove that $\KS(x\cnd y)\le E_U(x,y)+O(1)$, and the same for $\KS(y\cnd x)$. This is obvious since the definition of $E_U$ contains more requirements for $p$ (it should map both $x$ to $y$ and $y$ to $x$, while in $\KS(x\cnd y)$ it is enough to map $y$ to $x$).

To prove the reverse inequality, consider for each $n$ the binary relation $R_n$ on strings (of all lengths) defined as
$$
R_n(x,y)\Leftrightarrow \KS(x\cnd y)< n \text{ and } \KS(y\cnd x) <n.
$$
By definition, this relation is symmetric. It is easy to see that $R_n$ is (computably) enumerable uniformly in $n$, since we may compute better and better upper bounds for $\KS$ reaching ultimately its true value. We think of $R_n$ as the set of edges of an undirected graph whose vertices are binary strings. Note that each vertex $x$ of this graph has degree less than $2^n$ since there are less than $2^n$ programs of length less than $n$ that map $x$ to its neighbors.

For each $n$, we enumerate edges of this graph (i.e., pairs in $R_n$). We want to assign colors to the edges of $R_n$ in such a way that edges that have a common endpoint have different colors. In other terms, we require that for every vertex $x$ all edges of $R_n$ adjacent to $x$ have different colors. For that, $2^{n+1}$ colors are enough. Indeed, each new edge needs a color that differentiates it from less than $2^n$ existing edges adjacent to one its endpoint and less than $2^n$ edges adjacent to other endpoint.

Let us agree to use $(n+1)$-bit strings as colors for edges in $R_n$, and perform  this coloring in parallel for all $n$. Now we define $U(p,x)$ for a $(n+1)$-bit string $p$ and arbitrary string $x$ as the string $y$ such that the edge $(x,y)$ has color $p$ in the coloring of edges from $R_n$. Note that $n$ can be reconstructed as $|p|-1$. The uniqueness property for colors guarantees that there is at most one $y$ such that $(x,y)$ has color $p$, so $U(p,x)$ is well defined. It is easy to see now that if $\KS(x\cnd y)<n$ and $\KS(y\cnd x)<n$, and $p$ is the color of the edge $(x,y)$, then $U(p,x)=y$ and $U(p,y)=x$ at the same time. This implies the reverse inequality (the $O(1)$ terms appears when we compare our $U$ with the optimal one).
\end{proof}

\begin{remark}
In the definition of information distance given above we look for a program $p$ that transforms $x$ to $y$ and also transforms $y$ to $x$. Note that we \emph{do not tell the program which of the two transformations is requested}. A weaker definition would provide also this information to $p$. This modification can be done in several ways. For example, we may require in the definition of $E$ that $U(p,0x)=y$ and $U(p,1y)=x$, using the first input bit as the direction flag. An equivalent approach is to use two computable functions $U$ and $U'$ in the definition and require that $U(p,x)=y$ and $U'(p,y)=x$. This corresponds to using different interpreters for both directions.

It is easy to show that the optimal functions $U$ and $U'$ exist for this (two-interpreter) version of the definition. A priori we may get a smaller value of information distance  in this way (the program's task is easier when the direction is known, informally speaking). But it is not the case for the following simple reason. Obviously, this new quantity is still an upper bound for both conditional complexities $\KS(x\cnd y)$ and $\KS(y\cnd x)$ with $O(1)$ precision. Therefore Theorem~\ref{thm:bglvz} guarantees that this new definition of information distance coincides with the old one up to $O(1)$ additive terms. (For the prefix versions of information distance such a simple argument does not work anymore, see below.)
\end{remark}

We have seen that different approaches lead to the same (up to $O(1)$ additive term) notion of plain information distance. There is also a simple and natural quantitative characterization of this notion provided by the following theorem.

\begin{theorem}
The function $E_U$ for optimal $U$ is the minimal up to $O(1)$ additive terms upper semicomputable non-negative symmetric function $E$ with two string arguments and natural values such that 
$$
  \#\{y\colon E(x,y)<n\} \le c2^n \eqno(*)
$$
for some $c$ and for all integers $n$ and strings $x$.  
\end{theorem}

Recall that upper semicomputability of $E$ means that one can compute a sequence of total upper bounds for $E$ that converges to $E$. The equivalent requirement: the set of triples $(x,y,n)$ where $x,y$ are strings and $n$ are natural numbers, such that $E(x,y)<n$, is (computably) enumerable.

\begin{proof}
The function $\max(\KS(x\cnd y),\KS(y\cnd x))$ is upper semicomputable and symmetric. The inequality $(*)$ is true for it since it is true for the smaller function $\KS(y\cnd x)$ (for $c=1$; indeed, the number of programs of length less than $n$ is at most $2^n$).

On the other hand, if $E$ is some symmetric upper semicomputable function that satisfies $(*)$, then one can for any given $x$ and $n$ enumerate all $y$ such that $E(x,y)<n$. There are less than $c2^n$ strings $y$ with this property, so each $y$ can be described (given $x$) by a string of $n+\lceil \log c\rceil$ bits, its ordinal number in the enumeration. Note that the value of $n$ can be reconstructed from this string (by decreasing its length by $\lceil \log c\rceil$), so $\KS(y\cnd x)\le n+O(1)$ if $E(x,y)<n$. It remains to apply the symmetry of $E$ and Theorem~\ref{thm:bglvz}.
\end{proof}

\begin{remark}
The name ``information distance'' motivates the following question: does the plain information distance satisfy the triangle inequality? For the logarithmic precision the answer is positive, because 
$$
\KS(x\cnd z)\le \KS(x\cnd y)+\KS (y\cnd z)+O(\log (\KS(x\cnd y)+\KS(y\cnd z))).
$$
However, if we replace the last term by an $O(1)$-term, then the triangle inequality is no more true. Indeed, for every strings $x$ and $y$ the distance between an empty string $\Lambda$ and $x$ is $\KS(x)+O(1)$, and the distance between $x$ and some encoding of a pair $(x,y)$ is at most $\KS(y)+O(1)$, and the triangle inequality for distances with $O(1)$-precision would imply $\KS(x,y)\le\KS(x)+\KS(y)+O(1)$, and this is not true, see, e.g., \cite[section 2.1]{usv}.
\end{remark}

One may ask whether a weaker statement saying that there is a maximal (up to an $O(1)$ additive term) function in the class of all symmetric non-negative functions $E$ that satisfy both the condition $(*)$ and the triangle inequality, is true. The answer is negative, as the following proposition shows.

\begin{proposition}\label{prop:triangle_inequality}
There are two upper semicomputable symmetric functions $E_1$, $E_2$ that both satisfy the condition $(*)$ and the triangle inequality, such that no function that is bounded both by $E_1$ and $E_2$ can satisfy $(*)$ and the triangle inequality at the same time.
\end{proposition}

\begin{proof}
Let us agree that $E_1(x,y)$ and $E_2(x,y)$ are infinite when $x$ and $y$ have different lengths. If $x$ and $y$ are $n$-bit strings, then $E_1(x,y)\le k$ means that all the bits in $x$ and $y$ outside the first $k$ positions are the same, and $E_2(x,y)\le k$ is defined in a symmetric way (for the last $k$ positions). Both $E_1$ and $E_2$ satisfy the triangle inequality (and even the ultrametric inequality) and also satisfy the condition $(*)$, since the ball of radius $k$ consist of strings that coincide except for the first/last $k$ bits. If $E$ is bounded both by $E_1+O(1)$ and $E_2+O(1)$ and satisfies the triangle inequality, then by changing the first $k$ and the last $l$ positions in a string $x$ we get a string $y$ such that $E(x,y)\le k+l$, and it is easy to see that the number of strings that can be obtained in this way is not $O(2^{k+l})$, but $\Theta((k+l)2^{k+l})$.
\end{proof}

\section{Prefix complexity: different definitions}\label{sec:prefix}

The notion of prefix complexity was introduced independently by Levin~\cite{levin1971,levin1974,gacs1974} and later by Chaitin~\cite{chaitin1975}. There are several versions of this definition, and they all turn out to be equivalent, so people usually do not care much about technical details that are different. However, if we want to consider the counterparts of these definitions for information distance, the difference becomes important if we are interested in $O(1)$-precision. 

Essentially there are four different definitions of prefix complexity that appear in the literature.

\subsection{Prefix-free definition}

A computable partial function $U(p,x)$ with two string arguments and string values is called \emph{prefix-free} (with respect to the first argument) if $U(p,x)$ and $U(p',x)$ cannot be defined simultaneously for a string $p$ and its prefix $p'$ \emph{and for the same second argument $x$}. In other words, for every string $x$ the set of strings $p$ such that $U(p,x)$ is defined is prefix-free, i.e., does not contain a string and its prefix at the same time.

For a prefix-free function $U$ we may consider the complexity function $\KS_U(y\cnd x)$. In this way we get a smaller class of complexity functions (compared with the definition of plain complexity discussed above), and the Solomonoff--Kolmogorov theorem can be easily modified to show that there exists a minimal complexity function in this smaller class (up to $O(1)$ additive term, as usual). This function is called \emph{prefix conditional complexity} and usually is denoted by $\KP(y\cnd x)$. It is greater than $\KS(y\cnd x)$ since the class of available functions $U$ is more restricted; the relation between $\KS$ and $\KP$ is well studied (see, e.g.,~\cite[chapter 4]{usv} and references within).

The unconditional prefix complexity $\KP(x)$ is defined in the same way, with $U$ that does not have a second argument. We can also define $\KP(x)$ as $\KP(x\cnd y_0)$ for some fixed string $y_0$. This string may be chosen arbitrarily; for each choice we have $\KP(x)=\KP(x\cnd y_0)+O(1)$ but the constant in the $O(1)$ bound depends on the choice of $y_0$.

\subsection{Prefix-stable definition}

The prefix-stable version of the definition considers another restriction on the function~$U$. Namely, in this version the function $U$ should be \emph{prefix-stable} with respect to the first argument. This means that if $U(p,x)$ is defined, then $U(p',x)$ is defined and equal to $U(p,x)$ for all $p'$ that are extensions of $p$ (i.e., when $p$ is a prefix of $p'$). We consider the class of all computable partial prefix-stable functions $U$ and corresponding functions $\KS_U$, and observe that there exists an optimal prefix-stable function $U$ that makes $\KS_U$ minimal in this class (for prefix-stable functions).

It is rather easy to see that the prefix-stable definition leads to a version of complexity that does not exceed the prefix-free one (each prefix-free computable function can be easily extended to a prefix-stable one). The reverse inequality is not so obvious and there is no known direct proof; the standard argument compares both versions with the third one (the logarithm of a maximal semimeasure, see Section~\ref{subsec:semimeasure} below for this definition). 

Prefix-free and prefix-stable definitions correspond to the same intuitive idea: the program should be ``self-delimiting''. This means that the machine gets access to an infinite sequence of bits that starts with the program and has no marker indicating the end of a program. The prefix-free and prefix-stable definitions correspond to two possible ways of accessing this sequence. The prefix-free definition corresponds to a blocking read primitive (if the machine needs one more input bit, the computation waits until this bit is provided). The prefix-stable definition corresponds to a non-blocking read primitive (the machine has access to the input bits queue and may continue computations if the queue is currently empty). We do not go into details here; the interested reader could find this discussion in~\cite[section 4.4]{usv}.

\subsection{A priori probability definition}\label{subsec:apriori}

In this approach we consider the \emph{a priori probability} of $y$ given $x$, the probability of the event ``a random program maps $x$ to $y$''. More precisely, consider a prefix-stable function $U(p,x)$ and an infinite sequence $\pi$ of independent uniformly distributed random bits (a random variable). We say that $U(\pi,x)=y$ if $U(p,x)=y$ for some $p$ that is a prefix of $\pi$. Since $U$ is prefix-stable, the value $U(\pi,x)$ is well defined. For given $x$ and $y$, we denote by $m_U(y\cnd x)$ the probability of this event (the measure of the set of $\pi$ such that $U(\pi,x)=y$). For each prefix-stable $U$ we get some function $m_U$. It is easy to see that there exists an optimal $U$ that makes $m_U$ maximal (up to an $O(1)$-factor). Then we define prefix complexity $\KP(y\cnd x)$ as $-\log m_U(y\cnd x)$ for this optimal $U$. 

It is also easy to see that prefix-free functions $U$ (used instead of prefix-stable ones) lead to the same definition of prefix complexity. Informally speaking, if we have an infinite sequence of random bits as the first argument, we do not care whether we have blocking or non-blocking read access, the bits are always there. The non-trivial and most fundamental result about prefix complexity is that this definition (as logarithm of the probability) is equivalent to the two previous ones. As a byproduct of this result we see that the prefix-free and prefix-stable definitions are equivalent. This proof and the detailed discussion of the difference between the definitions can be found, e.g., in~\cite[chapter 4]{usv}.

\subsection{Semimeasure definition}\label{subsec:semimeasure}

The semimeasure approach defines a priori probability in a different way, as a convergent series that converges as slow as possible. More precisely, a \emph{lower semicomputable semimeasure} is a non-negative real-valued function $m(x)$ on binary strings such that $m(x)$ is a limit of a computable (uniformly in $x$) increasing sequence of rational numbers and $\sum_x m(x)\le 1$. There exists a maximal (up to $O(1)$-factor) lower semicomputable semimeasure $\mm(x)$, and its negative logarithm coincides with (unconditional) prefix complexity $\KP(x)$ up to an $O(1)$ additive term. 

We can define conditional prefix complexity in the same way, considering semimeasures with parameter $y$. Namely, we consider lower semicomputable non-negative real-valued functions $m(x,y)$ such that  $\sum_x m(x,y)\le 1$ for every $y$. Again there exists a maximal function among them, denoted by $\mm(x\cnd y)$, and its negative logarithm equals $\KP(x\cnd y)$ up to an $O(1)$ additive term.

To prove this equality, we note first that the a priori conditional probability $m_U(x\cnd y)$ is a lower semicomputable conditional semimeasure. The lower semicomputability is easy to see: we can simulate the machine $U$ and discover more and more programs that map $y$ to $x$. The inequality $\sum_x m_U(x\cnd y)$ also has a simple probabilistic meaning: the events ``$\pi$ maps $y$ to $x$'' for a given $y$ and different $x$ are disjoint, so the sum of their probabilities does not exceed $1$. The other direction (starting from a semimeasure, construct a machine) is a bit more difficult, but in fact it is possible (even exactly, without additional $O(1)$-factors). See~\cite[chapter 4]{usv} for details.

The semimeasure definition can be reformulated in terms of complexities (by taking exponents): $\KP(x\cnd y)$ is a minimal (up to $O(1)$ additive term) upper semicomputable non-negative integer function $k(x,y)$ such that
$$
\sum_x2^{-k(x,y)} \le 1
$$ 
for all $y$.
A similar characterization of plain complexity would use a weaker requirement
$$
\#\{x\colon k(x,y) < n\} < c2^n
$$
for some $c$ and all $y$. (We discussed a similar result for information distance where the additional symmetry requirement was used, but the proof is the same.)

\subsection{Warning}\label{subsec:warning}

There exists a definition of plain conditional complexity that does \emph{not} have a prefix-version counterpart. Namely, the plain conditional complexity $\KS(x\cnd y)$ can be equivalently defined as the \emph{minimal unconditional plain complexity of a program that maps $y$ to $x$}.  In this way we do not need the programming language used to map $y$ to $x$ to be optimal; it is enough to assume that we can computably translate programs in other languages into our language; this property, sometimes called \emph{$s$-$m$-$n$-theorem} or \emph{G\"odel property of a computable numbering}, is true for almost all reasonable programming languages. Of course, we still assume that the language used in the definition of unconditional Kolmogorov complexity is optimal. 

One may hope that $\KP(x\cnd y)$ can be similarly defined as the minimal (unconditional) prefix complexity of a program that maps $y$ to $x$.  The following proposition shows that it is not the case.

\begin{proposition}\label{prop:not-program-complexity}
The prefix complexity $\KP(x\cnd y)$ does not exceed the minimal prefix complexity of a program that maps $y$ to $x$; however, the difference between these two quantities is not bounded.
\end{proposition}

\begin{proof} To prove the first part, assume that $U_1(p)$ is a prefix-stable function of one argument that makes the complexity function 
$$
\KS_{U_1}(q)=\min\{ |p|\colon U(p)=q\}
$$
 minimal. Then $\KS_U(q)=\KP(q)+O(1)$. (We still need an $O(1)$ term since the choice of an optimal prefix-stable function is arbitrary). Then consider the function
$$
U_2(p,x)=[U_1(p)](x)
$$
where $[q](x)$ denotes the output of a program $q$ on input $x$. Then $U_2$ is a prefix-stable function from the definition of conditional prefix complexity, and 
$$
\KS_{U_2}(y\cnd x) \le \KS_{U_1}(q)
$$
for any program $q$ that maps $x$ to $y$ (i.e., $[q](x)=y$). This gives the inequality mentioned in the proposition. Now we have to show that this inequality is not an equality with $O(1)$-precision.

Note that $\KP(x\cnd n)\le n+O(1)$ for every binary string $x$ of length $n$. Indeed, a prefix-stable (or prefix-free) machine that gets $n$ as input can copy 
$n$ first bits of its program to the output. (The prefix-free machine should check that there are exactly $n$ input bits.) In this way we get $n$-bit programs for all strings of length $n$. 

Now assume that the two quantities coincide up to an $O(1$) additive term. Then for every string $x$ there exists a program $q_x$ that maps $|x|$ to $x$ and $\KP(q_x)\le |x|+c$ for all $x$ and some~$c$. Note that $q_x$ may be equal to $q_y$ for $x\ne y$, but this may happen only if $x$ and $y$ have different lengths. Consider now the set $Q$ of all $q_x$ for all strings $x$, and the series 
$$
\sum_{q\in Q} 2^{-\KP(q)}.\eqno(**)
$$
This sum does not exceed $1$ (it is a part of a similar sum for all $q$ that is at most $1$, see above). On the other hand, we have at least $2^n$ different programs $q_x$ for all $n$-bit strings $x$, and they correspond to different terms in $(**)$; each of these terms is at least $2^{-n-c}$. We get a converging series that contains, for every $n$, at least $2^n$ terms of size at least $2^{-n-c}$. It is easy to see that such a series does not exist. Indeed, each tail of this series should be at least $2^{-c-1}$ (consider these $2^n$ terms for large $n$ when at least half of these terms are in the tail), and this is incompatible with convergence. 
\end{proof}

Why do we get a bigger quantity when considering the prefix complexity of a program that maps $y$ to $x$? The reason is that the prefix-freeness (or prefix-stability) requirement for the function $U(p,x)$ is formulated separately for each $x$: the decision where to stop reading the program $p$ \emph{may depend on its input}~$x$. This is not possible for a prefix-free description of a program that maps $x$ to~$y$. It is easy to overlook this problem when we informally describe prefix complexity $\KP(x\cnd y)$ as ``the minimal length of a program, written in a self-delimiting language, that maps $y$ to $x$'', because the words ``self-delimiting language''  implicitly assume that we can determine where the program ends while reading the program text (and before we know its input), and this is a wrong assumption. 

\subsection{Historical digression}\label{subsec:history}

Let us comment a bit on the history of prefix complexity. It appeared first in 1971 in Levin's PhD thesis~\cite{levin1971}; Kolmogorov was his thesis advisor. Levin used essentially the semimeasure definition (formulated a bit differently). This thesis remained unpublished for a very long time (and it was in Russian). In 1974 G\'acs' paper~\cite{gacs1974} appeared where the formula for the prefix complexity of a pair was proven. This paper mentioned prefix complexity as ``introduced by Levin in [4], [5]'' (\cite{levin1973} and \cite{levin1974} in our numbering).  The first of these two papers does not say anything about prefix complexity explicitly, but defines the monotone complexity of sequences of natural numbers, and prefix complexity can be considered as a special case when the sequence has length $1$ (this is equivalent to the prefix-stable definition of prefix complexity).  The second paper (we discuss it later in this section) has a comment ``(to appear)'' in G\'acs' paper.

G\'acs does not reproduce the definition of prefix complexity saying only that it is ``defined as the complexity of specifying $x$ on a machine on which it is impossible to indicate the endpoint [the English translation says ``halting'' instead of ``endpoint'' but this is an obvious translation error] of a master program: an infinite sequence of binary symbols enters the machine and the machine must itself decide how many binary symbols are required for its computation''. This description is not completely clear, but it looks more like a prefix-free definition (if we understand it in such a way that the program is written on a one-directional tape and the machine decides where to stop reading).  G\'acs also notes that prefix complexity (he denotes it by $KP(x)$) ``is equal to the [negative] base two logarithm of a universal semicomputable probability measure that can be defined on the countable set of all words''. 

Levin's 1974 paper~\cite{levin1974} says that ``the quantity $KP(x)$ has been investigated in details in [6,7]''. Here [7] in Levin's numbering is G\'acs paper cited above (\cite{gacs1974} is our numbering) and has the comment ``in press'', and [6] in Levin's numbering is cited as 
[Levin L.A., On different version of algorithmic complexity of finite objects, to appear]. Levin does not have a paper with exactly this title, but the closest approximation is his 1976 paper~\cite{levin1976}, where prefix complexity is defined as the logarithm of a maximal semimeasure.  Except for these references, \cite{levin1974} describes the prefix complexity in terms of prefix-stable functions: ``It differs from the Kolmogorov complexity measure $\langle\ldots\rangle$ in that the decoding algorithm $A$ has the following ``prefix'' attribute: if $A(p_1)$ and $A(p_2)$ are defined and distinct, then $p_1$ cannot be a beginning fragment of $p_2$''.  

The prefix-free and a priori probability definitions were given independently by Chaitin in~\cite{chaitin1975} (in different notation) together with the proof of their equivalence, so~\cite{chaitin1975} was the first publication containing this (important) proof.
 
Now it seems that the most popular definition of prefix complexity is the prefix-free one (it is given as the main definition in~\cite{lv}, for example).

\section{Prefix complexity and information distance}\label{sec:prefix-distance}

\subsection{Four versions of prefix information distance}\label{subsec:four-versions}

Both the prefix-free and prefix-stable versions of prefix complexity have their counterparts for the information distance.

Let $U(p,x)$ be a partial computable prefix-free [prefix-stable] function of two string arguments having string values. Consider the function 
$$
E_U(x,y) = \min\{ |p|\colon U(p,x) = y \text{ and } U(p,y)=x\}
$$
As before, one can easily prove that there exists a minimal (up to $O(1)$) function among all functions $E_U$ of the class considered.  It will be called \emph{prefix-free} [resp.~\emph{prefix-stable}] \emph{information distance}  function. We clarify the difference between these variants.

Note that only the cases when $U(p,x)=y$ and also $U(p,y)=x$ matter for $E_U$. So we may assume without loss of generality that $U(p,x)=y \Leftrightarrow U(p,y)=x$ waiting until both equalities are true before finalizing the values of $U$. Then for every $p$ we have some matching $M_p$ on the set of all strings: an edge $x$--$y$ is in $M_p$ if $U(p,x)=y$ and $U(p,y)=x$. This is indeed a matching: for every $x$ only $U(p,x)$ may be connected with $x$.

The set $M_p$ is enumerable uniformly in $p$. In the prefix-free version the matchings $M_p$ and $M_q$ are disjoint (have no common vertices) for two compatible strings $p$ and $q$ (one is an extension of the other). For the prefix-stable version $M_p$ increases when $p$ increases (and remains a matching).  It is easy to see that a family $M_p$ that has these properties always corresponds to some function $U$ (here we have two statements: for prefix-free and prefix-stable version).

There is another way in which this definition could be modified. As we have discussed for the plain complexity, we may consider two different functions $U$ and $U'$ and consider the distance function
$$
E_{U,U'}(x,y) = \min\{ |p|\colon U(p,x) = y \text{ and } U'(p,y)=x\}.
$$
Intuitively this means that we know the transformation direction in addition to the input string. This corresponds to matchings in a bipartite graph where both parts consist of all binary strings; the edge $x$--$y$ is in the matching $M_p$ if $U(p,x)=y$ and $U'(p,y)=x$.  Again instead of the pair $(U,U')$ we may consider the family of matchings that are disjoint (for compatible $p$, in the prefix-free version) or monotone (for the prefix-stable version). In this way we get two other versions of information distance that could be called \emph{bipartite prefix-free} and \emph{bipartite prefix-stable} information distances.

In~\cite{bglvz} the information distance is defined as the prefix-free information distance (with the same function $U$ for both directions, not two different ones). The definition (section III) considers the minimal function among all $E_U$. This minimal function is denoted by $E_0(x,y)$ (while $\max(\KP(x\cnd y),\KP(y\cnd x))$ is denoted by $E_1(x,y)$, see section I of the same paper). The inequality $E_1\le E_0$ is obvious, and the reverse inequality (with logarithmic precision) is proven in~\cite{bglvz} as Theorem 3.3.

Which of the four versions of prefix information distance is the most natural?  Are they really different? It is easy to see that the prefix-stable version (bipartite or not) does not exceed the corresponding prefix-free version, since every prefix-free function has a prefix-stable extension. Also each bipartite version (prefix-free or prefix-stable) does not exceed the corresponding non-bipartite version for obvious reasons (one may take $U=U'$). It is hard to say which version is most natural, and the question whether some of them coincide or all four are different, remains open. But as we will see (Theorem~\ref{th:main}), the smallest of all four, the prefix-stable bipartite version, is still bigger than $E_1$ (the maximum of conditional complexities), and the difference is unbounded, so for all four versions (including the prefix-free non-bipartite version used both in~\cite{bglvz,lzlm,mahmud}) the equality with $O(1)$-precision is not true, contrary to what is said in~\cite{mahmud}.

However, before going to this negative result, we prove some positive results about the definition of information distance that is a counterpart of the a priori probability definition of prefix complexity.

\subsection{A priori probability of going back and forth}\label{subsec:apriori-distance}

Fix some prefix-free function $U(p,x)$. The conditional a priori probability $m_U (y\cnd x)$ is defined as
$$
\Pr_{\pi} [U(\pi,x)=y]
$$
where $U(\pi,x)=y$ means that $U(p,x)=y$ for some $p$ that is a prefix of $\pi$. As we discussed, there exists a maximal function among all $m_U$, and its negative logarithm equals the conditional prefix complexity $\KP(y\cnd x)$.

Now let us consider the counterpart of this construction for the information distance. The natural way to do this is to consider the function
$$
e_{U}(x,y)=\Pr_{\pi} [U(\pi,x)=y \text{ and } U(\pi,y)=x].
$$
Note that in this definition the prefixes of $\pi$ used for both computations are not necessarily the same. It is easy to show, as usual, that there exists an \emph{optimal} machine $U$ that makes $e_U$ maximal. Fixing some optimal $U$, we get some function $\ee(x,y)$ (different optimal $U$ lead to functions that differ only by $O(1)$-factor). The negative logarithm of this function coincides with $E_1$ (from~\cite{bglvz}) with $O(1)$-precision, as the following result says.

\begin{theorem}\label{thm:apriori}
$$
-\log \ee(x,y)=\max(\KP(x\cnd y),\KP(y\cnd x))+O(1).
$$
\end{theorem}

\begin{proof}
Rewriting the right-hand side in the exponential scale, we need to prove that 
$$
\ee(x,y)=\min (\mm(x\cnd y),\mm(y\cnd x))
$$
up to $O(1)$-factors. One direction is obvious: $\ee(x,y)$ is smaller than $\mm(x\cnd y)$ since the set of $\pi$ in the definition of $\ee$ is a subset of the corresponding set for $\mm$, if we use the probabilistic definition of $\mm=m_U$. The same is true for $\mm(y\cnd x)$. 

The non-trivial part of the statement is the reverse inequality. Here we need to construct a machine $U$ such that 
$$
e_U(x,y)\ge \min (\mm(x\cnd y),\mm(y\cnd x))
$$
up to $O(1)$-factors.

Let us denote the right-hand side by $u(x,y)$. The function $u$ is symmetric, lower semicomputable and $\sum_y u(x,y)\le 1$ for all $x$ (due to the symmetry, we do not need the other inequality where $y$ is fixed). This is all we need to construct $U$ with the desired properties; in fact $e_U(x,y)$ will be at least $0.5u(x,y)$ (and the factor $0.5$ is important for the proof).

Every machine $U$ has a ``dual'' representation: for every pair $(x,y)$ one may consider the subset $U_{x,y}$ of the Cantor space that consists of all $\pi$ such that $U(\pi,x)=y$ and $U(\pi,y)=x$. These sets are effectively open (i.e., are computably enumerable unions of intervals in the Cantor space) uniformly in $x,y$, are symmetric ($U_{x,y}=U_{y,x}$) and have the following property: for a fixed $x$, all sets $U_{x,y}$ for all $y$ (including $y=x$) are disjoint.

What is important to us is that this correspondence works in both directions. If we have some family $U_{x,y}$ of uniformly effectively open sets that is symmetric and has the disjointness property mentioned above, there exists a prefix-free machine $U$ that generates these sets as described above. This machine works as follows: given some $x$, it enumerates the intervals that form $U_{x,y}$ for all $y$ (it is possible since the sets $U_{x,y}$ are effectively open uniformly in $x,y$). One may assume without loss of generality that all the intervals in the enumeration are disjoint. Indeed, every effectively open set can be represented as a union of a computable sequence of disjoint intervals (to make intervals disjoint, we represent the set difference between the last interval and previously generated intervals as a a finite union of intervals). Note also that for different values of~$y$ the sets $U_{x,y}$ are disjoint by the assumption. If the enumeration for $U_{x,y}$ contains the interval $[p]$ (the set of all extensions of some bit strings $p$), then we let $U(p,x)=y$ and $U(p,y)=x$ (we assume that the same enumeration is used for $U_{x,y}$ and $U_{y,x}$). Since all intervals are disjoint, the function $U(p,x)$ is prefix-free.

Now it remains (and this is the main part of the proof) to construct the family $U_{x,y}$ with required properties in such a way that the measure of $U_{x,y}$ is at least $0.5u(x,y)$. In our construction it will be \emph{exactly} $0.5u(x,y)$. For that we use the same idea as in~\cite{bglvz} but in the continuous setting. Since $u(x,y)$ is lower semicomputable, we may consider the increasing sequence $u'(x,y)$ of approximations from below (that increase with time, though we do not explicitly mention time in the notation) that converge to $u(x,y)$. We assume that at each step one of the values $u'(x,y)$ increases by a dyadic rational number $r$. In response to that increase, we add to $U_{x,y}$ one or several intervals that have total measure $r/2$ and do not intersect $U_{x,z}$ and $U_{z,y}$ for any $z$. For that we consider the unions of all already chosen parts of $U_{x,z}$ and of all chosen parts of $U_{z,y}$. The measure of the first union is bounded by $0.5\sum_z u'(x,z)$ and the measure of the second union is bounded by $0.5\sum_z u'(z,y)$ where $u'$ is the lower bound for $u$ before the $r$-increase. Since the sums remain bounded by $1$ after the $r$-increase, we may select a subset of measure $r/2$ outside both unions. (We may even select a subset of measure $r$, but this will destroy the construction at the following steps, so we add only $r/2$ to $U_{x,y}$.)
\end{proof}

\begin{remark}
As for the other settings, we may consider two functions $U$ and $U'$ and the probability of the event 
$$
e_{U,U'}(x,y)=\Pr_{\pi} [U(\pi,x)=y \text{ and } U'(\pi,y)=x]
$$
for those $U,U'$ that make this probability maximal. The equality of Theorem~\ref{thm:apriori} remains valid for this version. Indeed, the easy part can be proven in the same way, and for the difficult direction we have proven a stronger statement with additional requirement $U=U'$. 
\end{remark}

One can also describe the function $\ee$ as a maximal function in some class, therefore getting a quantitative definition of $E_0$. This is essentially the statement of theorem 4.2 in~\cite{bglvz}. In terms of semimeasures it can be reformulated as follows.

\begin{proposition}\label{prop:qpd}
Consider the class of symmetric lower semicomputable functions $u(x,y)$ with string arguments and non-negative real values such that $\sum_y u(x,y)\le 1$ for all $x$. This class has a maximal function that coincides with $\min(\mm(x\cnd y), \mm(y\cnd x))$ up to an $O(1)$ factor.
\end{proposition}

(Indeed, we have already seen that this minimum has the required properties; if some other function $u(x,y)$ in this class is given, we compare it with conditional semimeasures $\mm(x\cnd y)$ and $\mm(y\cnd x)$ and conclude that $u$ does not exceed both of them.)

In logarithmic scale this statement can be reformulated as follows: \emph{the class of upper semicomputable symmetric functions $D(x,y)$ with string arguments and real values such that $\sum_y 2^{-D(x,y)}\le 1$ for each $x$, has a minimal element that coincides with $\max(\KP(x\cnd y),\KP(y\cnd x))$ up to an $O(1)$ additive term}. Theorem 4.2 in~\cite{bglvz} says the same with the additional condition for $D$: it should satisfy the triangle inequality. This restriction makes the class smaller and could increase the minimal element in the class, but this does not happen since the function
$$
\max (\KP(x\cnd y),\KP(y\cnd x))+c
$$
satisfies the triangle inequality for large enough $c$. This follows from the inequality $\KP(x\cnd z)\le \KP(x\cnd y)+\KP(y\cnd z)+O(1)$ since the left hand size increases by $c$ and the right hand size increases by $2c$ when $\KP$ is increased by $c$.

\begin{remark}
To be pedantic, we have to note that in~\cite{bglvz} an additional condition $D(x,x)=0$ is required for the functions in the class; to make this possible, one has to exclude the term $2^{-D(x,x)}$ in the sum (now this term equals $1$) and require that $\sum_{y\ne x} 2^{-D(x,y)}\le 1$ (p.~1414, the last inequality). Note that the triangle inequality remains valid if we change $D$ and let $D(x,x)=0$ for all $x$.
\end{remark}

\section{A counterexample}\label{sec:game}

In this section we present the main negative (and most technically difficult) result of this paper that shows that none of the four prefix distances coincides with $$E_1(x,y)=\max(\KP(x\cnd y),\KP(y\cnd x)).$$

\begin{theorem}\label{th:main}
The bipartite prefix-stable information distance exceeds $E_1(x,y)$ more than by a constant: the difference is unbounded.
\end{theorem}

As we have mentioned, the other three versions of the information distance are even bigger, so the same result is true for all of them. We will explain the proof for the non-bipartite prefix-stable version (it is a bit easier and less notation is needed) and then explain the changes needed for the bipartite prefix-stable version.
Our proof also provides a lower bound in terms of the length: for strings of length $n$, the difference can be as large as
\[
\log \log n - O(\log \log \log n).
\]

The proof uses the game approach (see~\cite{muchnik-vyugin,muchnik-vereshchagin} for the general context, but the proof is self-contained). In the next section (\ref{subsec:game-enough}) we explain the game rules and prove that a computable winning strategy in the game implies that the difference is unbounded, and then (in Section~\ref{subsec:game-strategy}) we explain the strategy. Finally (in Section~\ref{subsec:game-bipartite}) we discuss the modifications needed for the bipartite case.

\section{Equality if the distance is superlogarithmic}

Given the previous result, all distances become equal for pairs of strings of equal length, provided their distance is not too small.

\begin{theorem}\label{th:informationDistanceEquality}
  If $|x|=|y|$ and $E_1(x,y) \ge 6\log |x|$, then all four prefix information distances are equal to $E_1(x,y)+O(1)$.
\end{theorem}

This seems to be the first equality in information theory whose precision becomes smaller if the quantity becomes larger.
The proof is given in the second appendix.


\appendix
\section{Proof of Theorem~\ref{th:main} }

\subsection{It is enough to win a game}\label{subsec:game-enough}

Consider the following two-player full information game. Fix some parameter $c$, a positive rational number. The game field is the complete graph on a countable set (no loops); we use binary strings as graph vertices.  Alice and Bob take turns. 

Alice increases \emph{weights} of the graph edges. We denote the weight of the edge connecting vertices $u$ and $v$ by $m_{u,v}$ (here $u\ne v$). Initially all $m_{u,v}$ are zeros. At each move Alice may increase weights of finitely many edges using rational numbers as new weights. The weights should satisfy the inequality $\sum_{v\ne u} m_{u,v}\le 1$ for every $u$ (the total weight of the edges adjacent to some vertex should not exceed~$1$).

Bob assigns some subsets of the Cantor space to edges. For each $u,v$ (where $u\ne v$) the set $M_{u,v}$ assigned to the edge $u$--$v$ is a clopen subset of the Cantor space (clopen subsets are subsets that are closed and open at the same time, i.e., finite unions of intervals in the Cantor space). Initially all $M_{u,v}$ are empty. At each move Bob may increase sets assigned to finitely many edges (using arbitrary clopen sets that contain the previous ones). For every $u$, the sets $M_{u,v}$ (for all $v\ne u$) should be disjoint.

The game is infinite, and the winner is determined in the limit (assuming that both Alice and Bob follow the rules). Namely, Bob wins if for every $u$ and $v$ (where $u\ne v$) the limit value $\lim M_{u,v}$ (the union of the increasing sequence of Bob's labels for edge $u$--$v$) contains an interval in the Cantor space whose size is at least $c\cdot \lim m_{u,v}$ (the limit value of Alice's labels for $u$--$v$, multiplied by $c$). Recall that the interval $[z]$ in the 
Cantor space is the set of all extensions of some string~$z$, and its size is $2^{-|z|}$. In the sequel the size of the maximal interval contained in $X$ is denoted by $\nu(X)$.

We claim that the existence of a computable (uniformly in $c$) winning strategy for Alice in this game is enough to prove Theorem~\ref{th:main}. But first let us make some remarks on the game rules.

\begin{remark}
Increasing the constant $c$, we make Bob's task more difficult, and Alice's task easier. So our claim says that Alice can win the game even for arbitrarily small (positive) values of $c$.
\end{remark}

\begin{remark}
In our definition the result of the game is determined by the limit values of $m_{u,v}$ and $M_{u,v}$, so both players may postpone their moves. Two consequences of this observation will be used. First, we may assume that Bob always has empty $M_{u,v}$ when $m_{u,v}=0$ (he may 
postpone his move). Second, we may assume that Bob has to satisfy the requirement $\nu(M_{u,v})\ge cm_{u,v}$ after each of his moves. Indeed, Alice may wait until this requirement is satisfied by Bob: if this never happens, Alice wins the game in the limit (due to compactness: if an infinite family of intervals covers some large interval in the Cantor space, a finite subfamily exists that covers it, too).
\end{remark}

Now let us assume that Alice has a (uniformly) computable strategy for winning the game for every $c>0$. Since the factor $c$ is arbitrary, we may strengthen the requirement for Alice and require $\sum_{v\ne u} m_{u,v}\le d$ for some $d>0$. This corresponds to the factor $cd$ in the original game. Given some integer $k>0$, consider Alice's winning strategy for $c=2^{-k}$ and $d=2^{-k}$. We  play all these strategies simultaneously 
against a ``blind'' strategy for Bob that ignores Alice's moves and just follows the optimal machine $U$ used in the definition of information distance. Here are the details.

Consider the function $U$ that makes the function 
$$
E_U(u,v) = \min\{ |p|\colon U(p,u) = v \text{ and } U(p,v)=u\}
$$
minimal. For each edge $u$--$v$ consider the union of the sets $[p]$ for all $p$ such that $U(p,u)=v$ and $U(p,v)=u$ at the same time. This union is an effectively open set, and Bob enumerates the corresponding intervals and adds 
them to the label for the edge $u$--$v$ when they appear in the enumeration. (Note that this set is the same for $(u,v)$ and $(v,u)$ by definition.) For the limit set $M_{u,v}$ we then have $\nu(M_{u,v})\ge 2^{-E_U(u,v)}$ by construction (consider the interval that corresponds to the shortest $p$ in the definition of $E_U(u,v)$).

Let Alice use her winning strategy (for $c=2^{-k}$ and $d=2^{-k}$) against Bob. Since Bob's actions and Alice's strategy are computable, the limit values of Alice's weights are lower semicomputable uniformly in $k$. Let us denote these limit values by $m^k_{u,v}$ (for the $k$th game). We know that for every $u$ and $k$ the sum $\sum_{v\ne u} m^k_{u,v}$ does not exceed $2^{-k}$. Therefore the sum 
$$
 m_{u,v} = \sum_k m^k_{u,v}
$$
satisfies the requirement
$$
\sum_{v\ne u} m_{u,v} \le 1
$$
and we can apply Proposition~\ref{prop:qpd}, where we let $m_{u,u}=0$. This proposition guarantees that
$$
m_{u,v} \le O(\min(\mm(u\cnd v),\mm(v\cnd u))=2^{-E_1(u,v)+O(1)}.
$$
If, contrary to the statement of Theorem~\ref{th:main}, the value of the prefix-stable (non-bipartite) information distance between $u$ and $v$ is bounded by $E_1(u,v)+O(1)$, then $E_1(u,v)$ in the right hand side of the last inequality can be replaced by $E_U(u,v)$. But this means, by our construction, that Bob wins the $k$th game for large enough $k$, since the maximal intervals in $M_{u,v}$ are large enough to match $m_{u,v}$ (and therefore $m^k_{u,v}$) for large enough $k$, according to this inequality. We get a contradiction that finishes the proof of Theorem~\ref{th:main} for the non-bipartite case, assuming the existence of a uniformly computable winning strategy for Alice.

\subsection{How to win the game}\label{subsec:game-strategy}

Now we present a winning strategy for Alice. It is more convenient to consider an equivalent version of the game where Alice should satisfy the requirement $\sum m_{u,v}\le d$ (for some constant $d$; we assume without loss of generality that $d$ is a negative power of $2$)  and Bob should match Alice's weights without any factor, i.e., satisfy the requirement $\nu(M_{u,v})\ge m_{u,v}$. We need to show that even for small values of $d$ Alice has a winning strategy.

The idea of the strategy is that Alice maintains a finite set of ``currently active'' vertices, initially very large and then decreasing. The game is split into $N$ stages where $N=2/d$ (as we will see, this is enough). After each stage the set of active vertices and the edge labels satisfy the following conditions:
\begin{itemize} 
\item Alice has zero weights on edges that connect active vertices (as we have said, we may assume without loss of generality that Bob has empty labels on these edges, too);
\item for each active vertex, only a small weight is used by Alice on edges that connect it to other vertices (inactive ones; edges to active ones are covered by the previous condition and do not carry any weight); this weight will never exceed $d/2$;
\item more and more space is unavailable to Bob for use on edges between active vertices, since it is already used on edges connecting active and inactive vertices. 
\end{itemize}
The amount of unavailable space (for Bob) grows from stage to stage until no more space is available and Alice wins. In fact, at each stage the amount of unavailable space grows by $d/2$, so Alice needs $N=2/d$ stages to make all the space unavailable for Bob; then she makes one more request and wins since Bob has no available space to fulfull this request.

In the previous paragraph we used the words ``unavailable space'' informally. What do we mean by unavailable space? Consider some active vertex $x$ and edges that connect it to inactive vertices. These edges have some of Bob's labels (subsets of the Cantor space). The part of the Cantor space occupied by these labels is not available to Bob for edges between $x$ and other active vertices. Moreover, if Alice requests an interval of size $\eps$, and some part (even a small one) of an interval of this size is occupied, then this interval cannot be used by Bob (is unavailable). In this way the unavailable space can be much bigger than the occupied space, and this difference is the main tool in our argument.\footnote{This type of accounting goes back to G\'acs' paper~\cite{gacs1983} where he proved that monotone complexity and continuous a priori complexity differ more than by a constant, see also~\cite{usv} for the detailed exposition of his argument.} 

Let us explain this technique. First, let us agree that Alice increases only zero weights, and the new non-zero weight
she uses depends on the stage only. At the first stage she uses some very small $\eps_0$, at the second stage she uses some bigger $\eps_1$, etc. (so at the $i$th stage weights $\eps_{i-1}$ are used). We will use values of $\eps_i$ that are powers of~$2$ (since interval sizes in the Cantor space are powers of~$2$ anyway), and assume that $\eps_0\ll \eps_1\ll \eps_2\ldots$. More precisely, we let $\eps_N=d/2$ and assume that $\eps_{i-1}/\eps_i = d/2$. 
\begin{center}
\includegraphics[scale=1]{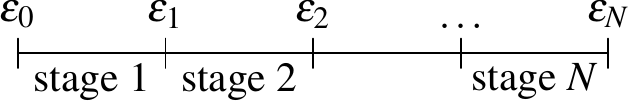}
\end{center}
This commitment about the weights implies that, starting from the $(i+1)$th stage, only the $\eps_i$-neighborhood of the space used by Bob matters. Here by $\eps$-neighborhood (where $\eps$ is a negative power of $2$) of a subset $X$ of the Cantor space we mean the union of all intervals of size $\eps$ that have nonempty intersection with $X$; note that the $\eps$-neighborhood of $X$ increases when $\eps$ increases (or $X$ increases).  

More precisely, let us call an interval \emph{dirty for active vertex~$x$} (at some moment) if some part of this interval already appears in Bob's labels for edges that connect $x$ to inactive vertices. This interval cannot be used later by Alice. After stage $i$, we consider all the intervals of size $\eps_i$ that are ``everywhere dirty'', i.e., dirty for all vertices (those that are dirty for some active vertices but not for the others, do not count). The everywhere dirty intervals form the \emph{unavailable space after stage $i$}, and the total measure of this space increases at least by $d/2$ at each stage. In other terms, after stage $i$ we consider for every active vertex $x$ the space allocated by Bob to all edges connecting $x$ with (currently) inactive vertices, and the $\eps_i$-neighborhood of this space. The intersection of these neighborhoods for all active vertices $x$ is the unavailable space (after stage $i$). 

After stage $i$ the total size of unavailable space will be at least $i/N$ (recall that $N=2/d$). At the end (after the $N$th stage) we have $\eps_N=d/2$, so the total size of everywhere dirty intervals of size $d/2$ is $N/N=1$, while the total weight used by Alice at any vertex is $d/2$. Then Alice makes one more request with weight $d/2$ and wins. Of course, we need that at least two vertices remain active after stage $N$, and this will be guaranteed if the initial number of active vertices is large enough.

The picture above places $\eps_i$ between stages since $\eps_i$ is used for accounting after stage $i$ and before stage $i+1$.

\medskip

It remains to explain how Alice plays at stage $i$ using requests of size $\eps_{i-1}$ and creating (new) everywhere dirty intervals of size $\eps_i$ with total size (=the size of their union) at least $d/2$. This happens in several substages; each substage decreases the set of active vertices and increases the set of everywhere dirty intervals of size $\eps_i$ (for the remaining active vertices).

Before starting each substage, we look at two subsets of the Cantor space:
\begin{itemize}
\item[(a)] the set of intervals of size $\eps_{i-1}$ that were everywhere dirty after the previous stage;
\item[(b)] the set of intervals of size $\eps_i$ that are everywhere dirty now (after the substages that are already performed).
\end{itemize}
The second set is bigger for two reasons. First, we changed the granularity (recall the $\eps_i$-neighborhood of some set can be bigger than $\eps_{i-1}$-neighborhood). Second, the previous substages create new everywhere dirty intervals of size $\eps_i$. Our goal is to make the second set larger than the first one; the required difference in size is $d/2$. If this goal is already achieved, we finish the stage (no more substage are necessary). If not, we initiate a new substage that creates a new everywhere dirty $\eps_i$-interval.

\begin{center}
\vbox{%
\begin{center}%
\includegraphics[scale=1]{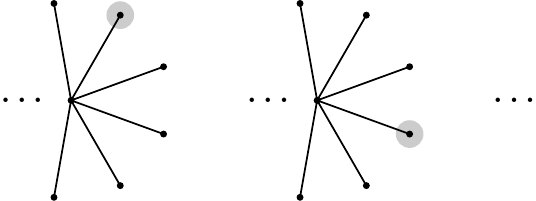}\\[1ex]
Alice's strategy for a substage
\end{center}%
}%
\end{center}

The key idea is that if Alice makes requests for all edges of a large star, she may use a lot of weight for the central vertex (the sum of the weights could be up to $d/2$, since in our process the total weight on edges that connect some active vertex to inactive ones, never exceeds $d/2$, and the maximal total weight is $d$). Still for all other vertices of the star only one new edge of non-zero weight $\eps_{i-1}$ is added, and the central vertex will be made inactive after the substage. Bob has to allocate some intervals of size at least $\eps_{i-1}$ for every edge in the star, and these intervals should be disjoint (due to the restrictions for the center of the star). 
The total measure of these intervals is at least $d/2$, and all of them are outside the zone (a). Therefore, since the goal is not yet achieved, one of these new intervals used by Bob is also outside the zone (b). 

Alice does the same for many stars (assuming that there are enough active vertices) and gets many new $\eps_i$-intervals outside the (b)-zone (one per star). Some of them have to coincide: if we started with many stars, we may select many new active vertices that have the same new $\eps_i$-dirty interval. Making all other vertices inactive, we get a smaller (but still large if we started with a large set of active vertices) set of active vertices and a new everywhere dirty $\eps_i$-interval. The goal of a substage is achieved, and we may look again at the set of everywhere dirty $\eps_i$-intervals (with new interval added) to decide whether the difference between (b) and (a) is now enough ($d/2$) or a new substage is needed. The maximal number of substages needed to finish the stage is $(d/2)/\eps_i$, since each substage creates a new $\eps_i$-interval.

The same procedure is repeated for all $N$ stages. We need to check that Alice does not violate her obligation on the sum of weights connecting some active vertex to all inactive vertices. For that, we look at the ``amplification factor'': in the construction Alice uses a new weight $\eps_{i-1}$ (for every new active vertex) to get a dirty interval of size $\eps_i$, therefore the amplification factor is $\eps_i/\eps_{i-1}=2/d$. Since the total size of dirty intervals is at most $1$, the total weight used by Alice (for each active vertex) never exceeds $d/2$, as required.

It remains to explain why Alice can choose enough active vertices in the beginning, so she will never run out of them in the construction and at least two active vertices exist at the end (so the last request $d/2$ for the edge connecting them wins the game). Indeed, the backwards induction shows that for each substage of each stage there is some finite number of active vertices that is sufficient for Alice to follow her plan till the end. If we want to upper bound the length on the strings where a given difference between two quantities in the statement of Theorem~\ref{th:main} is  achieved, we need to compute this number explicitly. But the qualitative statement (the unbounded difference) is already proven for the prefix-stable non-bipartite case. The prefix-free case is a corollary (the distance becomes bigger), but for the bipartite case we need to adapt the argument, and this is done in the next section.

\subsection{Modifications for the bipartite case}\label{subsec:game-bipartite}

In the bipartite case the game should be changed. Namely, we have a complete bipartite graph where left and right parts contain all strings. Alice increases weights on edges; for each vertex (left or right) the sum of the weights for all adjacent edges should not exceed some $d$ (the parameter of the game). In other terms, at each step Alice's weights form a two-dimensional table $m_{x,y}$ indexed by pairs of strings $x$ and $y$, all entries are zeros except for finitely many positive rational numbers, and 
$$
\forall x\, \left(\sum_y m_{x,y}\le 1\right), \quad
\forall y\, \left(\sum_x m_{x,y}\le 1\right)
$$
(now we have two requirements since the table is not symmetric anymore; note that the diagonal entries $m_{x,x}$ do not have special status).

Bob replies by assigning increasing sets $M_{x,y}$ to edges such that $\nu(M_{x,y})\ge m_{x,y}$. For each $x$ the sets $M_{x,y}$ (with different $y$) should be disjoint; the same should be true for sets $M_{x,y}$ for fixed $y$ and different $x$.

Again, to prove that the bipartite prefix-free information distance exceeds $E_1(x,y)=\max(\KP(x\cnd y),\KP(y\cnd x))$ by a constant, we show that for every $d$ Alice has a computable (uniformly in $d$) winning strategy in this game. Then we consider games with total weight $2^{-k}$ and factor condition $\nu(M_{x,y})\ge 2^{-k}m_{x,y}$ and let Alice play her winning strategy against the ``blind'' strategy for Bob that (for the edge $x$--$y$) enumerates all intervals $[p]$ such that $U(p,x)=y$ and $U'(p,y)=x$ at the same time.

The winning strategy for Alice works in stages as before, and the request size grows with the stage in the same way.  Alice keeps the list of active vertices (both in the left and the right part), and after each stage (and substage) all weights on edges between (left and right) active vertices are zeros, and sum of Alice's weights on edges between each active vertex and all inactive vertices is small. After the $i$th stage we consider intervals of size $\eps_i$. When defining dirty intervals, we look only at one part (say, the right one). An interval $I$ of size $\eps_i$ is considered as dirty for a right vertex $y$ if some part of this interval is allocated to some edge connecting $y$ to some vertex $x$ from the left part. We are interested in intervals that are dirty everywhere (i.e., for every right vertex $y$). At each substage (of the $i$th stage), to create a new everywhere dirty interval, we use stars as shown. 
\begin{center}
\includegraphics[scale=1.2]{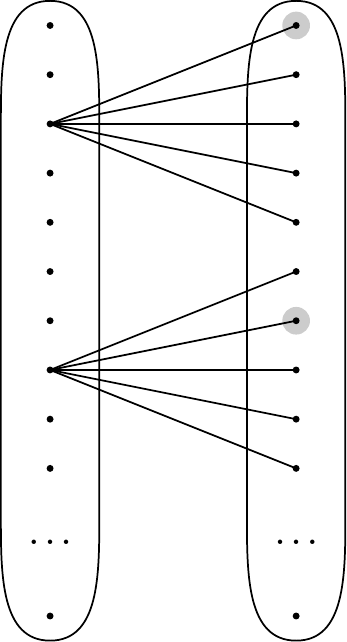}
\end{center}
In each star the sum of Alice's weights is $d/2$; we choose an edge for which Bob's label is not in the everywhere dirty intervals found during previous substage, and look at the $\eps_i$ interval 
where it goes. Since there are many stars, some dirty interval occurs many times; Alice selects such an interval and uses the right vertices of corresponding edges as new active right vertices. On the left side, Alice uses vertices that do not appear in the stars, as new active left vertices. (In this way we have much more active vertices on the left; if for some reason we want to keep the same number of left and right vertices, Alice may delete part of the remaining vertices..)

\section{Proof of  Theorem~\ref{th:informationDistanceEquality} }

\subsection*{Easy strategies for Bob}

We use the same game as before, but this time, we only consider strings of length exactly~$n$.
The game has 2 parameters, $n$ and $d>0$.
Alice must satisfy $m_{u,v} \le d$, and Bob can use all of his Cantor spaces, 
in other words, we set~$c=1$. This time, we must provide a winning strategy for Bob 
that works for some fixed constant $d>0$ and all~$n$. 

Each time Alice increases a weight $m_{u,v}$ to a value~$\varepsilon$, 
we say that she makes a {\em request} of the form~$(\{u,v\},\varepsilon)$. 
To each string $u$, we associate a Cantor space $\Omega_u$. 
If Bob replies by enumerates an interval $[q]$ in $M_{u,v}$, we say that he allocates 
$[q]$ both in $\Omega_u$ and $\Omega_v$. 
Recall that prefix-free non-bipartite information distance is maximal. 
Hence, we assume that Alice's requests are negative powers of $2$ and 
that Bob must always reply by assigning an interval of at least the same size, unless 
he has already done this.

To understand the main ideas behind Bob's strategy, we first consider 
variants of the game that are easier for Bob to win.

\smallskip
In the most basic variant, we require that 
Alice can only make requests of a fixed size~$\varepsilon$. 
In this case, the game reduces to the argument we used for the plain 
information distance, and a greedy strategy works for all~$d \le 1/2$.

\smallskip
In a related variant, we assume that Alice only uses request sizes from a finite set $\Epsilon$.
We also assume that Alice and Bob are given a probability measure $P$ over~$\Epsilon$,
and that all values of $P$ and~$\Epsilon$ are dyadic.
For each string $u$ and $\varepsilon \in \Epsilon$, 
let $W_{u,\varepsilon}$ be the total measure of requests of size $\varepsilon$ on~$u$, 
in other words, it is equal to $\varepsilon$ times the number of requests of the form~$(\{u,\cdot\}, \varepsilon)$.
We consider the game in which Alice's requests should satisfy $W_{u,\varepsilon} \le d \cdot P(\varepsilon)$ for all $u$ and $\varepsilon \in \Epsilon$.

In this case a winning strategy for Bob exists if $d \le 1/2$:
he divides each Cantor space in regions, and to each size $\varepsilon \in \Epsilon$, he associates 
a region of measure~$P(\varepsilon)$. All Cantor spaces $\Omega_u$ are partitioned  in the same way.
When given a request of size $\varepsilon$, Bob uses the same greedy strategy as before inside the corresponding region.
For the same reason as before, this strategy works for all~$d \le 1/2$.

\medskip
Now, consider the variant in which $\Epsilon$ has size $s$, i.e., Alice can use $s$ different sizes, 
and she can choose $P$ during the game. In other words, 
 Alice's requests should satisfy: 
\[
\sum_{\varepsilon \in \Epsilon}  \max_{|u|=n} \{W_{\varepsilon,u}\}
\;\;\le\;\; d\,.
\]

Bob has a winning strategy for $d=1/4$. It goes as follows. 
He creates $2s$ regions of equal size in Cantor space. 
He assigns the first $s$ regions to elements of $\Epsilon$. 
Initially, the other $s$ regions remain unassigned.
For all strings $u$, he uses the same partition of $\Omega_u$ into regions.
Each time $\max_u W_{u,\varepsilon}$ exceeds half of the measure of the $\varepsilon$-regions,
a new region is assigned to~$\varepsilon$.

What is the maximal number $r$ of assigned regions that can appear?
If $r_\varepsilon$ is the number of assigned regions for some $\varepsilon$, 
then
$
d \cdot P(\varepsilon) \ge (r_\varepsilon - 1) \cdot \tfrac{1}{2} \cdot \tfrac{1}{2s}.
$
Summing over $\varepsilon$, this implies $d \ge (r-s)/(4s)$.
Thus at most~$r = 2s$ regions can be assigned, and hence, the strategy can always proceed. 

\bigskip
\noindent
In the above variant, we considered a finite $\Epsilon$, and we soon explain 
why this is not an important restriction if $\Epsilon$ has size~$s=n$.
However, in the last strategy, we assumed that the distribution 
over request sizes in $\Epsilon$ is somehow the same for all strings, 
and avoiding this restriction is the hardest part of the argument.

To understand that the last strategy fails in the general case, 
assume that Alice makes requests of $s$ different sizes, 
and fixes $s$ strings for which she makes only requests of a single size.
Then the previous strategy needs $s^2/2$ regions, but has only $2s$ regions available.

This suggests the following approach. She creates $r=2s$ regions by partitioning 
the Cantor space in a large (but still polynomial) number of blocks of equal size, 
and the regions are obtained as random unions of blocks. In other words, the partitions 
of Cantor space is obtained from a random partition of the set of blocks.
Each string receives its own random partition. 
The idea is to use only 1 size in each block, so that the greedy strategy in the blocks is sufficient.

We explain this idea in more detail. 
The regions are numbered by $1, 2, \dots, r$.
For each pairs of strings,  about a fraction $1/r$ 
of the blocks will be assigned to the same region. 
We assign sizes to regions, as in the previous example: if the regions for a size 
gets too full, we add an unassigned region to this size. 
We say that a block is non-full 
if less than half of its measure is assigned.
When given a request $(\{u,v\},\varepsilon)$, 
we choose a block that, both for~$u$ and for~$v$, is non-full and lies in an $\varepsilon$-region.
If such a block exists, we allocate an available interval in this block. 
Note that such interval exists, because inside blocks we can again play the trivial strategy.

The problem is that the required block might not exist. 
We will prove that in a random partition, this event does not happen very often 
and that there exists a different strategy that can handle all remaining requests.

In the next four subsections we present the proof.
First, we discuss a game in which Alice's moves are restricted, and show how a winning strategy for Bob implies the theorem.
Then, we present the strategy discussed above that handles the bulk of the requests. 
Afterwards, we present a second strategy that handles the remaining requests. 
Finally, some combinatorial statements are proven that imply correctness of the two above strategies.


\subsection*{The restricted game}

We consider the game as described above. Recall that $c=1$, is played on $n$-bit strings, and 
Alice's requests sizes are negative powers of~$2$. 
We additionally require that Alice's requests satisfy the following conditions.

\begin{enumerate}[wide,labelwidth=!,labelindent=0pt]
\item[(A)]
    {\em All requests have size at least $2^{-n}$.} 

\item[(B)]
    {\em All requests have size at most $n^{-p}$.} (We use this at the end with $p=6$.)
\end{enumerate}

\noindent
Assumption (A) implies that at most $n$ different request sizes are used.
If Bob has a winning strategy in a game with parameter $d$ in which Alice is restricted by (A), then
he also has a winning strategy in an unrestricted game with parameter at most~$d/2$, 
because Bob can start by connecting all pairs using the bottom half of the Cantor spaces. 
After this, he can ignore all small requests.

\smallskip
\noindent
Assumption (B) changes the situation significantly, because it prevents Alice from playing the 
winning strategy of the previous result. In the next subsections, we present Bob's winning strategy.
We now show that this implies Theorem~\ref{th:informationDistanceEquality}.

\begin{lemma}\label{lem:theoremToRestrictedGame}
  If there exists a $d>0$ such that Bob has a computable winning strategy in the restricted game for $p=6$ and for all $n$, 
  then  Theorem~\ref{th:informationDistanceEquality} is true.
\end{lemma}

\begin{proof}
  We construct a prefix-free machine $V$ such that for all different $n$-bit strings $x,y$ 
  there is a program $q$ for which $V(q,x) = y$, $V(q,y) = x$ and $|q| \le \Emax(x,y) + O(\log \tfrac{1}{d})$. 

  \medskip
  \noindent
  {\em Computation of $V$ on  input $(q,x)$.} 
  Bob's winning strategy for strings of length $n = |x|$ 
  is played against Alice's strategy in which she sets $m_{u,v} = d2^{-\Emax(u,v)}$,
  using a non-increasing approximation of~$\Emax$. 
  Note that when approximations improve, the weights increase.
  Bob replies by enumerating intervals $M_{u,v}$, and 
  if $[q]$ is enumerated in some set $M_{x,v}$, 
  the output is $V(q,x) = v$.

  \medskip
  \noindent
  Note that the winning condition indeed implies that for all different $n$-bit $x,y$, 
  there exists a string $q$ of length $\log \tfrac{1}{d} + \Emax(x,y)$ with $V(q,x) = y$ and $V(q,y) = x$.
  Also, 
  The prefix-free non-bipartite information distance defined using $V$ satisfies the conditions of the lemma.
\end{proof}

\subsection*{The first substrategy}

The first substrategy uses the top halves of all Cantor spaces~$\Omega_u$, and the second strategy the bottom halves. 
It allocates most requests $(\{x,y\},\varepsilon)$, and when it fails to allocate some request, 
it does something extra: it {\em blames} either $x$ or $y$. 
We show that:

\vspace*{-5mm}
\begin{flalign}\tag{*}\label{cond:blamedLittle}
  \text{\parbox{0.9\textwidth}{
  {\em Every string is blamed at most $O(n^3)$ times.
  } 
  }
  }
\end{flalign}
\vspace*{-5mm}

\noindent
Let $n$ be large.
The substrategy starts by partitioning the Cantor space in $\ell = 2^7 n^4$ blocks of equal size.
These blocks are assigned to $r=2n$ regions. For each string, this partitioning happens differently.
The assignment is represented by a colouring of the blocks using $r$ different colours from a set~$\Sigma$. 
Thus, each assignment corresponds to an element $v \in \Sigma^\ell$, and vice versa, 
 each such $v$ and $a \in \Sigma$ determine a region containing blocks with indices
\[
v[a] = \{i : v_i  \mathsmaller{\, =\, } a\}. 
\]
For each $n$-bit $x$, we use a list $v$ from a set $S \subseteq \Sigma^\ell$ that satisfies the conditions 
of the following lemma.

\newcommand{\combLemmaOne}{
  Let $|\Sigma| = r \ge 2$. If $n \ge 2 + \log r$ and $\ell \ge 2^{7} r^3n$, 
  there exists a set $S \subseteq \Sigma^\ell$ of size $2^n$
  such that for all $(v,a) \in S \times \Sigma$: 
  \begin{itemize}[leftmargin=*]
    \item 
      at least a $\tfrac{1}{2r}$-fraction of elements in $v$ are equal to $a$, i.e., $\big|v[a]\big| \;\ge\; \tfrac{\ell}{2r}$, 

    \item  for every $I \subseteq v[a]$ of size $\ell/(8r)$, there are at most $O(r^2)$ pairs $(w,b) \in S \times \Sigma$ for which
  \end{itemize}
  \[
  \moreHorizSpaceInFormulas
  \Big| I \cap w[b] \Big| \ge \tfrac{1}{2}\Big| v[a] \cap w[b] \Big|.
  \]
  }
\begin{lemma}\label{lem:expanderlike}
  \combLemmaOne
\end{lemma}

\noindent
We postpone the proof to the last subsection, and continue with the substrategy.
Before the first request arrives, the substrategy searches for a set $S$ that satisfies the conditions of 
the lemma. For each $n$-bit $x$, it determines the corresponding allocation of blocks into regions.
It assigns the request sizes $2^{-1}$, $2^{-2}$, \dots, $2^{-n}$ to the first $n$ regions.
These regions are called the {\em $\varepsilon$-active} regions for the corresponding request sizes~$\varepsilon$. 
To the other $n$ regions, no sizes are assigned, and they are inactive. 
For every request size, there will always be a unique active region.

Let $x$ be an $n$-bit string.
A block in $\Omega_x$ is {\em full} if at least half of its measure is allocated.
A region in $\Omega_x$ is  {\em full} if at least $\tfrac{1}{8}$-th of its blocks are full.

\medskip
{\em Allocating a request $(\{x,y\},\varepsilon)$.} 
The substrategy considers the $\varepsilon$-active regions of $x$ and $y$. 
It searches for a common block that is non-full for both strings.
If no such block exists, then no interval is assigned, 
and a string is blamed for which at least half of the common blocks are full, 
(ties can be broken in an arbitrary way, or both strings can be blamed).
If an $\varepsilon$-active region becomes full, it is made inactive, 
and a new $\varepsilon$-active region is assigned from the unused ones.

\medskip
Recall that $r = 2n$ and that only half of the Cantor space is used by this substrategy. 
By the first item of  Lemma~\ref{lem:expanderlike}, each region has size at least $1/(4r)$.
Thus, if $d \le \tfrac{n}{4r} \tfrac{1}{8} \tfrac{1}{2} = 2^{-7}$, at most $n$
regions can be full. Besides this, at most $n$ active regions are needed, and
hence, $r$ regions are enough in order to always allocate an unused region.

\medskip
The requirement \ref{cond:blamedLittle} follows from the second item of the lemma.
Indeed, for some string, let $I$ be the set of the full blocks at a given moment.
Each time a string is blamed while a region is active, at least half of the common blocks are inside $I$.
By the lemma, for a fixed region, there are at most $O(n^2)$ strings for which this can happen, and 
hence there are at most $O(n^2)$ requests for which the string can be blamed.
Because there are $2n$ regions, the maximal number of times a string can be
blamed is $O(n^3)$.  Hence, the requirement for the first substrategy are satisfied.

\subsection*{The second substrategy}

This strategy must allocate all remaining requests. Each string receives 2 types of requests. 

\medskip
\noindent
\textbf{Requests on which it is blamed.} 
The number of such requests is at most~$O(n^3)$.
Recall from assumption (B) that requests have size at most $n^{-p}$, and
hence, their measure is at most $O(n^{3-p})$. For large $p$, this can be made arbitrarily small.
We will again use regions and blocks, and 
will choose $p$ such that all these requests can fit in a small fraction of a single block 
(which has size about $1/n^2$). 

\medskip
\noindent
\textbf{Requests on which the string is not blamed.} 
Averaged over all strings, this number of requests is also $O(n^3)$. But for some 
strings the total measure can be a constant fraction of the Cantor space,  
but in any case, at most a fraction~$d$.

\medskip
The main idea is that since each string is blamed so rarely, 
it does not matter much where the corresponding allocations happen.
But the other requests need to be properly structured.
Thus in a request, the string that is not blamed will play a dominant role
by determining in which region the interval will be assigned.
The terminology is inspired by human relationships, where a person who can create 
blame (and guilt) in the other, typically has more influence in a decision.

We explain the initialization of the second substrategy,
which uses the remaining half of the Cantor spaces.
Again it is partitioned into $2n$ regions, but now the partition is straightforward: 
the regions have equal size, and are identical for all strings. 
Each region is further subdivided in blocks. We use 2 types of blocks:
the  {\em dominant} blocks are used to allocate requests in which the string is not blamed, 
and  {\em submissive} blocks, which are used for the other requests.
Each region is partitioned in $s = 64n$ equal blocks, and this partition is different for all strings.
For a fixed string, we use the same subdivision for all the regions 
(or different ones, it does not matter as long as the combinatorial requirements are met), 
and for each pair of strings, there should be a sufficient overlap 
between dominant blocks of the first string and submissive blocks of the second one.
For a fixed region, the division of each string can be represented as an $s$-bit string. 
If $a,b \in \{0,1\}^s$ represent divisions of the same region for two different strings, 
then the number of indices $i$ with $a_i = 1$ and $b_i=0$ 
equals~$\sum_{i=1}^s a_i (1-b_i)$, and should be at least~$s/8$.
The following lemma provides the required sets of partitions.

\newcommand{\combLemmaTwo}{
  For $s \ge 64 n$, 
  there exists a subset of $\{0,1\}^s$ of size $2^n$
  such that $\sum_{i=1}^{s} a_i (1-b_i) \ge s/8$ 
  for any two different elements $a$ and $b$ in the set. 
  }
\begin{lemma}\label{lem:pairwiseHammingDistances}
  \combLemmaTwo
\end{lemma}

\noindent
We postpone the proof to the next subsection.
We say that a dominant block is  {\em full} if 
the total measure of allocated intervals in it is at least half of its measure.
We call a region  {\em full} if at least $\tfrac{1}{8}$th of its dominant blocks are full.
As before, if a region is full, a new region is assigned for the measure and if
$d$ is small, a total of $2n$ regions is enough to execute the strategy.

For each request, the substrategy selects a block in the active region 
of the non-blamed string that is
(1) is  submissive for the blamed string, and 
(2) dominant and non-full for the string that is not blamed.
Inside this block, an unallocated interval is selected.

\medskip
We need to prove that this strategy always works under the assumption~\ref{cond:blamedLittle}.
Thus, we need to explain that a required block exists, and that it has a free interval.

The required block exists by the conditions of the lemma: for a given region in two different strings,
at least $\tfrac{1}{8}$th of the dominant blocks overlap with the submissive blocks of the other string.
Less than $\tfrac{1}{8}$ of these blocks are full for the dominant string, so a non-full block exists.

The existence of a free interval in this block follows by bounding the allocated 
measure of the block. By definition of non-full blocks, less than half 
of its measure is allocated by the string that is not blamed. 
The other string can have at most an allocated measure $O(n^{3-p})$, 
because as we discussed, this is the total measure of requests in which the string is not blamed, 
and since the block is submissive, this string has allocated only such requests.
If we choose $p = 6$, then for large $n$ this is less than half 
of the size of a single block (which is proportional to $1/n^2$).

We have proven that the given substrategy allocates all remaining requests.
To prove the theorem, it remains to prove the 2 combinatorial lemmas.

\subsection*{Combinatorial statements}

We start with Lemma~\ref{lem:pairwiseHammingDistances} from the previous subsection. 
It is a slightly stronger variant of the following known result:
for all $s$ sufficiently larger than $n$, there exist $2^n$ bitstrings of length $s$
whose pairwise Hamming distances are at least $s/8$. We restate the lemma.

\begin{lemma*}
  \combLemmaTwo
\end{lemma*}

\begin{proof}
  We use the probabilistic method.
  If we generate $a,b \in \{0,1\}^s$ randomly, 
  then each term $x_i(1-y_i)$ in the sum equals $1$ with probability~$\tfrac{1}{4}$. 
  The probability that the sum is smaller than $s/8$ is at most 
  \[
  \exp \left(-2 \frac{s}{8^2} \right)  \le \exp (-2n).
  \]
  Now assume we generate $2^n$ bitstrings randomly. 
  By the union bound, the probability 
  that the condition is violated for some pair of bitstrings is at most
  $
   2^{2n} \exp (-2n) < 1.
   $
  Hence, there must be at least 1 set of strings that satisfies the condition of the lemma.
\end{proof}

\bigbreak
\noindent
We now proceed to Lemma~\ref{lem:expanderlike}.

\begin{lemma*}
  \combLemmaOne
\end{lemma*}

\begin{proof} Again, we use the probabilistic method. 
  We need to show that a random set $S \subseteq \Sigma^\ell$ 
  satisfies the requirements of the lemma with positive probability.
  We show that the requirements are satisfied if for some $k \le O(r^2)$, neither of the following conditions hold:
  \begin{itemize}[leftmargin=*]
    \item 
      There exist two pairs $(v,a)$ and $(w,b)$ in  $S \times \Sigma$ such that either
      \[
      \big| v[a] \cap w[b] \big| \;<\; \tfrac{\ell}{2r^2} \qquad \text{or} \qquad
      \big| v[a] \cap w[b] \big| \;>\; \tfrac{2\ell}{r^2}\;.
      \]

    \item 
      There exist $(v,a), (w_1, b_1), \dots, (w_k,b_k) \in S \times \Sigma$ and $I \subseteq v[a]$ of size $\tfrac{\ell}{8r}$, 
      such that for all $i \le k$: 
      \[
      \moreHorizSpaceInFormulas
      \Big| I \cap w_i[b_i] \Big| \;\ge\; \frac{\ell}{4r^2}\,.
      \]
  \end{itemize}
  The first item of the lemma follows from the negation of the first condition after 
  summing the left possibility over all~$b \in \Sigma$.
  For later use, note that the right possibility implies that $|v[a]| \le 2\ell/r$.

  Note that by the negation of the first condition, $|I \cap w_i[b_i]| \ge \ell/(4r^2)$ implies the inequality of the second item of the lemma.
  Thus, if this happens for less than $k$ pairs, then the number of pairs $(w,b)$ in the second requirement is also less than~$k$. 

  \medskip
  We show that such that the probability of the first condition is true, is less than $\tfrac{1}{2}$, and that for some $k \le O(r^2)$, 
  the conditional probability of the second condition, given the negation of the first, is also less than $\tfrac{1}{2}$. Hence, 
  with positive probability none of the conditions are true, and this implies the existence of the required set.

  \medskip
   {\em The first condition.} 
  The probability that $v_i = w_i$ for random $v,w \in \Sigma^\ell$ is~$1/r$. 
  By the Chernoff bound, the probability that this deviates more than $\ell/(2r)$ is 
  at most $2\exp(-\frac{\ell}{2r^2})$. By the union bound, this happens for 
  2 pairs in $S \times \Sigma$ with probability at most
  \[
  2 \cdot (r2^{n})^2 \cdot \exp \left(-\frac{\ell}{2r^2}\right).
  \]
  This is less than $\tfrac{1}{2}$ if $2(n+ 2 + \log r)2r^2 \le \ell$, and this is true by assumption 
  on~$n$ and~$\ell$. 

  \medskip
   {\em The second condition.} 
   For a fixed $(v,a)$, $b_1$ and $I$ the quantity 
   $\left| I \cap w_1[b_1] \right|$ is the sum of $|I|$ independent 
   Bernoulli variables with~$p = \tfrac{1}{r}$. If the condition is satisfied, then
   \[
      \moreHorizSpaceInFormulas
      \sum_{i \le k} \left| I \cap w_i[b_i] \right| \ge \frac{k\ell}{4r^2} \, . 
      \]

  The first quantity is the sum of $k|I| = k\ell/8r$ independent Bernoulli variables with $p = \tfrac{1}{r}$. 
  By the Chernoff bound, this happens with probability at most $\exp(-2\frac{1}{(2r)^2}k|I|)$. 
  The probability that this happens 
  for some choice of $(v,a), (w_1, b_1), \dots, (w_k, b_k) \in S \times \Sigma$ and $I \subseteq v[a]$ is bounded by
  \[
  \left( r2^n \right)^{k+1} \; 2^{2\ell/r} \; \exp\left( -\frac{k\ell}{16 r^3} \right),
  \]
  where the second term bounds the number of different subsets $I \subseteq v[a]$ as $2^{|v[a]|}$, and is bounded using the 
  negation of the first condition.
  This expression is less than $\tfrac{1}{2}$ if 
  \begin{align*}
    (n + \log r + 1)(k+1) &\le \tfrac{1}{2} \; \frac{k\ell}{16 r^3} \\
    \frac{2\ell}{r} &\le \tfrac{1}{2} \; \frac{k\ell}{16r^3}.
  \end{align*}
  The first inequality follows from $k + 1 \le 2k$ and the assumptions on $n$ and $\ell$.
  The second follows by choosing $k \le O(r^2)$. The lemma is proven.
\end{proof}

\end{document}